\documentclass[letterpaper, 10 pt, conference]{ieeeconf}  

\IEEEoverridecommandlockouts                              
\usepackage{amsmath} 
\usepackage{amssymb}  

\usepackage{graphicx}

\newcommand{\eps}{\epsilon}

\newcommand{\R}{\mathbb{R}}
\newcommand{\Rbb}{\mathbb{R}}

\newcommand{\Ac}{\mathcal{A}}

\newcommand{\Sc}{\mathcal{S}}

\newcommand{\Pcal}{\mathcal{P}}

\newcommand{\Fcal}{\mathcal{F}}
\newcommand{\Ocal}{\mathcal{O}}

\def\simiid{\,{\buildrel {iid} \over \sim}\,}

\newcommand\braket[2]{\left\langle #1, #2 \right\rangle}
\newcommand\simplexProd{\Delta^{\Pcal_1} \times \dots \times \Delta^{\Pcal_I}}
\newcommand\parenth[1]{\left( #1\right)}
\newcommand\Exp[1]{\mathbb{E}\left[ #1\right]}

\newcommand{\adj}{\mathrm{Adj}}

\newcommand{\gauss}{\mathrm{Gauss}}
\DeclareMathOperator*{\argmin}{\arg\min}

\usepackage{amsthm}

\usepackage{algorithm}
\usepackage{algorithmic}

\newtheorem{assumption}{Assumption}
\newtheorem{definition}{Definition}
\newtheorem{proposition}{Proposition}
\newtheorem{lemma}{Lemma}
\newtheorem{theorem}{Theorem}
\newtheorem{corollary}{Corollary}

\title{\LARGE \bf
Differential Privacy of Populations in Routing Games
}

\author{Roy Dong, Walid Krichene, Alexandre M. Bayen, and S. Shankar Sastry
\thanks{R. Dong, W. Krichene, A. M. Bayen, and S. S. Sastry are with the Department of Electrical Engineering and Computer Sciences, 
        University of California, Berkeley, Berkeley, CA, 94707, USA
        {\tt\footnotesize $\{$roydong,walid,bayen,sastry$\}$@eecs.berkeley.edu}}%
}

\begin{document}

\maketitle
\thispagestyle{empty}
\pagestyle{empty}

\begin{abstract}

As our ground transportation infrastructure modernizes, the large amount of data being measured, transmitted, and stored motivates an analysis of the privacy aspect of these emerging cyber-physical technologies. In this paper, we consider privacy in the routing game, where the origins and destinations of drivers are considered private. This is motivated by the fact that this spatiotemporal information can easily be used as the basis for inferences for a person's activities. More specifically, we consider the differential privacy of the mapping from the amount of flow for each origin-destination pair to the traffic flow measurements on each link of a traffic network. We use a stochastic online learning framework for the population dynamics, which is known to converge to the Nash equilibrium of the routing game. We analyze the sensitivity of this process and provide theoretical guarantees on the convergence rates as well as differential privacy values for these models. We confirm these with simulations on a small example.

\end{abstract}

\section{Introduction}
\label{sec:intro}


With the decreasing cost and size of technologies, our ground transportation infrastructure is increasingly modernizing with new sensor systems, control algorithms, and actuation modalities. Although these technologies promise great gains in traffic performance, such as level of service or equity~\cite{Board2011}, an unprecedented amount of data is being measured, transmitted, and stored, and an analysis of the privacy aspect of this emerging cyber-physical technology is needed.

There have been a multitude of privacy conceptions in the philosophical and legal literatures. From an engineering perspective, the most commonly used paradigms are \emph{control over information} and \emph{secrecy}~\cite{Solove2002}.

In the abstract, control over information generally requires transparency to the person about what data is being collected and stored, consent to the transmission of this data to any parties, and an ability to correct mistakes in the data. As an example of how this conception works in practice, control over information forms the foundation of the Federal Trade Commission's Fair Information Practices.

On the other hand, secrecy focuses on which new inferences can be made about a person due to the information contained in the data; in this paradigm, a privacy breach occurs when there is a revelation of information that was previously not known, and the person felt that the information was private.

Throughout this paper, our conception of privacy will focus on the secrecy paradigm. In other words, we will focus on what new inferences can be made from the data collected by sensors in ground traffic infrastructures.

In the context of traffic systems, we consider the case where the origin and destination are considered private. This is motivated by the fact that this spatiotemporal information can easily used as the basis for inferences for a person's activities. For example, an executive at the carsharing company Uber claimed he could tell when its users were having an affair in a blog post~\cite{Tufekci2014}.

More specifically, we consider the differential privacy of the mapping from population sizes, i.e. the amount of flow for each origin-destination pair, to the traffic flow measurements on each link of a traffic network.

A popular modeling assumption is that the traffic flow is \emph{atomless}, i.e. a single vehicle cannot unilaterally affect the flows on links~\cite{sandholm2001potential,roughgarden2007,krichene2015learning}. This is designed to match our intuition that, under normal conditions, one vehicle does not contribute significantly to traffic.

However, this implies that, through our models, one vehicle has no effect on the traffic flow measurements. Thus, in this paper, we consider differential privacy with respect to population sizes: how much does traffic flow change when a non-negligible mass of vehicles switch origin-destination pairs?

This framework is applicable for when some aggregator wants to protect the privacy of several drivers. For example, Google can analyze how much it reveals about its users when it provides routes through Google Maps. Alternatively, companies can consider how much is revealed through their shipping patterns, since this detailed data can allow inferences about important business information, such as which consumer markets are being targeted, which companies are in the supply chain, and which locations have potential for future expansion.

To model the dynamics of the driver populations, we use an online learning model in which, at iteration $t$, each population chooses a distribution over its paths. The joint decision of all populations determines the flows over the edges of the network, which, in turn, determines the costs over paths. These costs are then revealed to the populations, and given this information, they can update their distributions. This online learning model has been applied to routing games in~\cite{blum2006routing}, where the authors show that any no-regret strategy is guaranteed to converge to an equilibrium. The same model is also used in~\cite{krichene2015MD}, where the authors show that if each population applies a mirror descent algorithm, the joint distribution converges to a Nash equilibrium.

Our contribution is an analysis of the differential privacy of the dynamics of the driver populations. In this article, we consider a stochastic version of the model in~\cite{krichene2015MD}, in which the populations only have access to a noisy measurement of the path costs. The presence of noise is essential in providing differential privacy, while still guaranteeing convergence to the equilibrium, using results from stochastic optimization~\cite{juditsky2011variational,krichene2015SMD}. 

The rest of the paper is organized as follows. In Section~\ref{sec:diff}, we review the engineering literature on privacy and develop some of the theory of differential privacy. In Section~\ref{sec:route_game}, we introduce the routing game in the context of privacy. In Sections~\ref{sec:smd} and \ref{sec:diff_priv_route}, we provide a learning model based on stochastic mirror descent. Here, we present theory on convergence rates and analyze the differential privacy of the routing game. 
In Section~\ref{sec:numerical}, we present a numerical example and we conclude in Section~\ref{sec:conclusion}.




%


\section{Differential privacy}
\label{sec:diff}


\subsection{Previous work}

Motivated by changing technologies, there has been a lot of recent research considering the issue of privacy. In this section, we will try to summarize the mathematical results in this line of research most relevant to this paper, noting both that the field is too rich for a comprehensive literature review and that privacy is a complicated social phenomenon of which a mathematical model is only one facet.

From a mathematical perspective, there have been several definitions of privacy. We seek to quickly survey a few definitions.

There has been work in inferential privacy, which seeks to bound the probability an adversary with a fixed set of information can correctly infer a hidden parameter, and uses a hypothesis testing model~\cite{Dong2014a}.

Additionally, there has been work in information-theoretic based definitions of privacy, which uses the mutual information between a private parameter and the publicly observable data~\cite{Sankar2011,Salamatian2013} or the conditional entropy of a private parameter given the observables~\cite{Venkitasubramaniam2013}.

Throughout this paper, we will focus on a definition of privacy first introduced in~\cite{Dwork2006}, called \emph{differential privacy}. This definition was originally designed for databases taking values in a finite alphabet, but has since been extended to consider the output of optimization algorithms~\cite{Duchi2012,Hsu2014,Huang2014a,Han2014} and dynamical systems~\cite{LeNy2014}. For a more detailed analysis of the interpretation of differential privacy, we refer the reader to~\cite{Dwork2014}.

Our work is closest to that in~\cite{Han2014}, where the authors considered the differential privacy of constraint sets in the context of gradient descent. Additionally, the work in~\cite{Duchi2012} is of relevance, as it provides several minimax bounds for stochastic mirror descent, considered in this paper.

\subsection{Theory}

In this section, we will formally define differential privacy, as well as present results needed in future sections.

First, let $(\Omega, \Ac, P)$ denote our underlying probability space. Also, let $\Theta$ be a set equipped with a symmetric binary relation $\adj$, called the \emph{adjacency} relation. The set $\Theta$ contains the possible values for a private parameter. Intuitively, the adjacency relation indicates which values should be roughly indistinguishable from the observable data. Although we never consider distributions or measures on $\Theta$, for brevity we will often treat $\Theta$ as a measurable space, where any subset of $\Theta$ is measurable. 

Furthermore, let $(S, \Sc)$ denote a measurable space and let $Y : \Theta \times \Omega \rightarrow S$ be a mapping such that $Y(\theta, \cdot)$ is measurable for every $\theta \in \Theta$. In other words, given $\theta$, $Y(\theta, \cdot)$ is a random element in $S$. For shorthand, we will write $Y_\theta$ to represent $Y(\theta,\cdot)$.

We can now present the definition of differential privacy.

\begin{definition}
	\emph{Differential privacy}: 
	We say a measurable mapping $Y : \Theta \times \Omega \rightarrow S$ is $(\eps,\delta)$-differentially-private if for all measurable sets $B \in \Sc$ and any $\theta, \theta' \in \Theta$ such that $\adj(\theta,\theta')$:
	\begin{equation}
		P(Y_\theta \in B) \leq \exp(\eps) P(Y_{\theta'} \in B) + \delta
	\end{equation}
	If $\delta = 0$, we will say this mapping is $\eps$-differentially-private.
\end{definition}

We note two consequences of this definition. The first lemma appears in~\cite{LeNy2014}.
\begin{lemma}
\label{lemma:expect}
	\emph{\cite{LeNy2014}}: 
	If a mapping $Y$ is $(\eps,\delta)$-differentially private then 
	\begin{equation}
		Eg(Y_\theta) \leq \exp(\eps) Eg(Y_{\theta'}) + \delta
	\end{equation}
	holds for all bounded measurable real-valued functions $g$ and all $\theta, \theta' \in \Theta$ such that $\adj(\theta,\theta')$.
\end{lemma}

The second lemma allows us to use tail bounds when analyzing differential privacy in certain contexts as we will see in Section~\ref{sec:smd}.
\begin{lemma}
\label{lemma:tail_event}
Fix some event $E$. Suppose $P(E) \geq 1 - \delta'$ and that, for all measurable sets $B \in \Sc$ and all $\theta, \theta' \in \Theta$ such that $\adj(\theta,\theta')$:
\[
P(\{Y_\theta \in B\} \cap E) \leq \exp(\eps) P(\{Y_{\theta'} \in B\} \cap E) + \delta
\]
Then, $Y$ is $(\eps, \delta + \delta')$-differentially-private.
\end{lemma}

\begin{proof}
Fix any measurable set $B \in \Sc$ and any adjacent $\theta, \theta' \in \Theta$. Then:
\begin{align*}
P(&Y_\theta \in B) \\
& = P( \{Y_{\theta} \in B\} \cap E) + P( \{Y_{\theta} \in B\} \cap E^c) \\
& \leq \exp(\eps) P( \{Y_{\theta'} \in B\} \cap E) + \delta + \delta'
\end{align*}
As desired.
\end{proof}

A result we will use in future sections is how differentially private mappings can be composed.

\begin{proposition}
	\label{prop:adapt}
	\emph{Adaptive composition}: 
	Suppose $Y_1 : \Theta \times \Omega \rightarrow S_1$ is $(\eps_1,\delta_1)$-differentially-private and $Y_2 : \Theta \times S_1 \times \Omega \rightarrow S_2$ is a measurable mapping such that $Y_2(\cdot, s, \cdot)$ is $(\eps_2, \delta_2)$-differentially-private for each fixed $s \in S_1$. Then the mapping $(\theta, \omega) \mapsto (Y_1(\theta,\omega),Y_2(\theta,Y_1(\theta,\omega),\omega))$ is $(\eps_1+\eps_2,\exp(\eps_2)\delta_1+\delta_2)$-differentially-private.
\end{proposition}

\begin{proof}
	Pick any set $A \in \Sc_1 \times \Sc_2$. Let $\mu_1(\theta,\cdot)$ denote the distribution of $Y_1(\theta)$ and $\mu_2(\theta,s,\cdot)$ denote the distribution of $Y_2(\theta,s)$. Furthermore, for $y \in S_1$, let $A_y = \{ y' \in S_2 : (y,y') \in A \}$ denote the slice of $A$ with respect to the first coordinate. Then, for any $\theta, \theta'$ such that $\adj(\theta,\theta')$:
	\begin{align*}
		P[ &( Y_1(\theta),Y_2( \theta, Y_1(\theta) ) ) \in A ] \\
		&= \int \mu_1(\theta,dy_1) P(Y_2(\theta, y_1) \in A_{y_1}) \\
		&\leq \int \mu_1(\theta,dy_1) [\exp(\eps_2) P(Y_2(\theta',y_1) \in A_{y_1}) + \delta_2 ] \\
		&= \exp(\eps_2) \int \mu_1(\theta,dy_1) [ P(Y_2(\theta',y_1) \in A_{y_1})] + \delta_2
	\end{align*}
	Let $g(y) = P(Y_2(\theta',y) \in A_{y})$, and note both that $g$ is a bounded, measurable function and $Eg(Y_1(\theta')) = P( (Y_1(\theta'),Y_2(\theta',Y_1(\theta'))) \in A)$. So, invoking Lemma~\ref{lemma:expect}:
	\begin{align*}
		P[ &( Y_1(\theta),Y_2( \theta, Y_1(\theta) ) ) \in A ] \\
		&\leq \exp(\eps_2) \int \mu_1(\theta,dy_1) [ P(Y_2(\theta',y_1) \in A_{y_1})] + \delta_2 \\
		&= \exp(\eps_2) E g(Y_1(\theta)) + \delta_2 \\
		&\leq \exp(\eps_2) [ \exp(\eps_1) E g(Y_1(\theta')) + \delta_1 ] + \delta_2 \\
		&= \exp(\eps_1 + \eps_2) P( (Y_1(\theta'),Y_2(\theta',Y_1(\theta'))) \in A) + \dots \\
		& \qquad \exp(\eps_2)\delta_1 + \delta_2
	\end{align*}
	As desired.
\end{proof}

Additionally, we can induct on Proposition~\ref{prop:adapt}. For brevity, we will sometimes write $Y_t(\theta,Y_1(\theta),\dots,Y_{t-1}(\theta),\cdot)$ simply as $Y_t(\theta)$.
\begin{corollary}
\label{cor:repeat}
	\emph{Repeated adaptive composition}: 
	Suppose $Y_1 : \Theta \times \Omega \rightarrow S_1$ is $(\eps_1,\delta_1)$-differentially-private and $Y_t : \Theta \times S_1 \times \dots S_{t-1} \times \Omega \rightarrow S_t$ is a measurable mapping such that $Y_t(\cdot,s_1,\dots,s_{t-1},\cdot)$ is $(\eps_t, \delta_t)$-differentially-private  for each fixed $(s_1,\dots,s_{t-1}) \in S_1 \times \dots \times S_{t-1}$ and $1 < t \leq T$.
	
	Then, the mapping $(\theta,\omega) \mapsto (Y_1(\theta),Y_2(\theta),\dots,Y_T(\theta))$ is $\left( \sum_{t = 1}^T \eps_t, \sum_{t = 1}^T \exp\left[ \sum_{t' = t+1}^T \eps_{t'} \right]\delta_t \right)$-differentially-private.
\end{corollary}

%
%
%

Finally, we note that the Gaussian distribution guarantees differential privacy.

\begin{definition}
	\emph{Sensitivity}: 
	The $\ell_2$ sensitivity of a function $f : \Theta \to \Rbb$ is given by:
	\begin{equation}
		\Delta_2 f = \sup_{\theta, \theta' \in \Theta : \adj(\theta,\theta')} \|f(\theta) - f(\theta')\|_2
	\end{equation}
\end{definition}

\begin{definition}
	
	The zero-mean Gaussian distribution on $\R$ with variance parameter $\sigma^2$, denoted $\gauss(\sigma^2)$, has the density
	\begin{equation}
		y \mapsto \frac{1}{\sqrt{(2\pi \sigma^2)}} \exp\left(\frac{-|y|^2}{2\sigma^2}\right)
	\end{equation}
	with respect to the Lebesgue measure.
\end{definition}



\begin{proposition}
\label{prop:gauss_mech}
	\emph{Gaussian mechanism~\cite{Dwork2014}}: 
	For $\eps \in (0,1)$, and $b^2 > 2 \ln(1.25/\delta)$, the mapping $Y_\theta = f(\theta) + Z$, where $Z_i \simiid \gauss(\sigma^2)$ for some $\sigma \geq b \Delta_2 f/\eps$, 
	 is $(\eps,\delta)$-differentially-private.
\end{proposition}


\section{The Routing Game}
\label{sec:route_game}


The routing game is given by:
\begin{itemize}
\item a directed graph $G = (V, E)$,
\item a set of non-decreasing, Lipschitz continuous edge cost functions $c_e: \Rbb_+ \to \Rbb_+$, $e \in E$,
\item a finite set of origin-destination pairs $(o_i, d_i) \in V \times V$, indexed by $i \in \{1, \dots, I\}$,
\item and a finite set of populations $P_k$, indexed by $k \in \{1, \dots, K\}$.
\end{itemize}
In a ground transportation setting, the nodes, i.e. elements in $V$, represent physical locations, and edges, i.e. elements in $E$, represent the roadways that connect two locations. The edge cost functions $c_e$ correspond to the amount of time taken when traveling along an edge $e$, and the non-decreasing assumption corresponds to the physical intuition that congestion worsens travel time. 
Finally, each population represents some aggregator that manages flows for all origin-destination pairs, such as Google or Waze. 

For a given origin-destination pair $(o_i, d_i)$, let $\Pcal_i$ be the set of simple paths connecting $o_i$ to $d_i$, and let $M_i \in \Rbb^{|E| \times |\Pcal_i|}$ be the edge-path incidence matrix, defined as follows:
\begin{equation}
\forall (e, p) \in E \times \Pcal_i, \ (M_i)_{e, p} = \begin{cases}
1 & \text{if $e \in p$} \\
0 & \text{otherwise.}
\end{cases}
\end{equation}

A population $P_k$ is given by a private vector $\theta_k \in \Rbb_+^{I}$, which specifies, for each origin-destination pair $(o_i, d_i)$, the total mass of traffic $(\theta_k)_i$ that belongs to this population, and that travels from $o_i$ to $d_i$. 
We assume there is some upper bound on the total size of the populations. Furthermore, we will define an adjacency relationship between private vectors.

\begin{assumption}
\label{ass:bound_theta}
It is common knowledge that $\theta$ is bounded. That is, there exists an $A_\theta < \infty$ such that, for every population $k$, $\| \theta_k \|_\infty \leq A_{\theta}$, and each population and outside observers know this bound.
\end{assumption}

\begin{definition}
\label{def:adj}
Two private parameters of populations $(\theta_k)_{k \in [K]}$ and $(\theta_k')_{k \in [K]}$ are \emph{adjacent} if there exists a $k^*$ such that $\theta_k = \theta_k'$ for $k \neq k^*$ and:
\[
\| \theta_{k^*} - \theta_{k^*}' \|_\infty \leq c
\]
\end{definition}

Recall that the adjacency relationship provides defines which pairs of private parameters should be roughly indistinguishable. Here, $c$ is a constant that will be determined by the populations, modeling the maximum amount that a single population can increase or decrease the flow in one origin-destination pair without having a significant effect on observable data.

The action set of population $P_k$ is a distribution vector $x_k \in \simplexProd$, where
\[
\Delta^{\Pcal_i} = \left\{m \in \Rbb_+^{|\Pcal_i|} : \sum_{p \in \Pcal_i} m_p = 1 \right\}
\]
is the set of probability distributions over $\Pcal_i$. In other words, every population chooses, for each origin-destination pair $(o_i, d_i)$, how to distribute its mass across the available paths $\Pcal_i$. For notational convenience, we will write $(x_k)_{\Pcal_i}$ to denote the sub-vector $((x_k)_p)_{p \in \Pcal_i} \in \Delta^{\Pcal_i}$, so that $x_k = ((x_k)_{\Pcal_1}, \dots, (x_k)_{\Pcal_I})$.

The flow allocations of all populations $(x_k)_{k \in [K]}$ determine the edge flows, defined as follows: the flow on edge $e$ is $\phi_e(x_1, \dots, x_K) = \sum_{k = 1}^K \sum_{i = 1}^I (\theta_k)_i \sum_{p \in \Pcal_i} (x_k)_p 1_{(e \in p)}$. The vector of edge flows can be written simply in terms of the incidence matrices:
\[
\phi(x_1, \dots, x_K) = \sum_{k=1}^K \sum_{i = 1}^I (\theta_k)_i M_i (x_k)_{\Pcal_i}
\]
The edge flows and edge costs determine the path costs. That is, the cost on path $p \in \Pcal_i$ is given by:
\[
\ell_p(x_1, \dots, x_K) = \sum_{e \in p} c_e(\phi_e(x_1, \dots, x_K))
\]
We will denote by $\ell_{\Pcal_i}(x_1, \dots, x_K)$ the vector $(\ell_{p}(x_1, \dots, x_K))_{p \in \Pcal_i}$, and $\ell = (\ell_{\Pcal_1}, \dots, \ell_{\Pcal_I}) \in \Rbb_+^{\Pcal_1} \times \dots \times \Rbb_+^{\Pcal_I}$.

Finally, the cost for population $P_k$ under distributions $x_1, \dots, x_K$ is
\[
\sum_{i = 1}^I (\theta_k)_i \sum_{p \in \Pcal_i} ((x_k)_{\Pcal_i})_p \ell_{p}(x_1, \dots, x_K)
\]
which we will denote, more concisely, as
\[
\braket{x_k}{\ell(x_1, \dots, x_K)}_{\theta_k}
\]
where we define the inner product as follows: for all $x, y \in \Rbb^{\Pcal_1} \times \dots \times \Rbb^{\Pcal_I}$
\begin{equation}
\braket{x}{y}_\theta = \sum_{i = 1}^I \theta_i \sum_{p \in \Pcal_i} x_p y_p
\end{equation}

\subsection{Nash equilibria and the Rosenthal potential function}
\begin{definition}
\label{def:nash}
A collection of population distributions $(x_k)_{k \in [K]}$ is a \emph{Nash equilibrium} (also called \emph{Wardrop equilibrium} in the traffic literature), if for every $k \in [K]$ and every $y \in \simplexProd$:
\[
\braket{x_k}{\ell(x_1, \dots x_K)}_{\theta_k} \leq \braket{y}{\ell(x_1, \dots x_K)}_{\theta_k}
\]
That is, no driver can improve their cost by unilaterally changing their path.
\end{definition}
Next, we show that the set of Nash equilibria of the game are exactly the set of minimizers of the Rosenthal potential, defined as follows:
\[
f(x_1, \dots, x_K) = \sum_{e \in E} \int_0^{\phi_e(x_1, \dots, x_K)} c_e(u)du 
\]
\begin{proposition}
\label{prop:potential_gradient}
The Rosenthal potential is convex, and its gradient with respect to $x_k$ is:
\[
\nabla_{x_k} f(x_1, \dots, x_K) = \sum_{i = 1}^I (\theta_k)_{i} \ell_{\Pcal_i}(x_1, \dots, x_K)
\]
\end{proposition}

%

\begin{corollary}
The set of Nash equilibria of the game is exactly the set of solutions of the following convex problem:
\begin{equation}
\label{eq:potential_problem}
\begin{aligned}
&\text{minimize} && f(x_1, \dots, x_K) \\
&\text{subject to }&& x_k \in \simplexProd \text{ for all } k \in [K]
\end{aligned}
\end{equation}
\end{corollary}

\section{Stochastic mirror descent dynamics and convergence to Nash equilibria}
\label{sec:smd}


\subsection{Online learning model}

We consider the following online learning model of the game: at each iteration $t \in \{1, 2, \dots, T\}$, every population $P_k$ chooses a distribution vector $x_k^{(t)} \in \simplexProd$. The combined choice of all populations determines the path loss vector $\ell(x^{(t)}_1, \dots, x_K^{(t)})$, which we will denote simply by $\ell^{(t)}$. The loss of population $k$ is then given by the inner product $\braket{\ell^{(t)}}{x^{(t)}_k}_{\theta_k}$.

At the end of iteration $t$, a stochastic loss vector $\hat \ell^{(t)}$, is revealed to all populations. Intuitively, one can think of $\hat \ell^{(t)}$ as a noisy version of $\ell^{(t)}$. The precise assumptions on the process $(\hat \ell^{(t)})$ will be given in Assumption~\ref{assumption:ell_hat}.

\subsection{Population dynamics}
Our population dynamics take the following form.
\begin{assumption}
We assume that for each population $P_k$, the stochastic process $(x_k^{(t)})$ follows stochastic mirror descent dynamics, given in Algorithm~\ref{alg:smd_dynamics}. 
\end{assumption}

These dynamics correspond to a stochastic version of the dynamics used in~\cite{krichene2015MD}.

\begin{algorithm}[t]
\begin{algorithmic}
{\small
\FOR{$t \in \{0, \dots, T-1\}$}
\STATE Observe $\hat \ell^{(t)}$
\STATE Update
\[
x_k^{(t+1)} = \argmin_{x_k \in \simplexProd} \braket{ \hat \ell^{(t)}}{x_k }_{\theta_k} + \frac{1}{\eta_k^{(t)}} D_{\psi_k}(x_k, x_k^{(t)})
\]
\ENDFOR%
}
\end{algorithmic}
\caption{Stochastic mirror descent dynamics for population $k$, with initial distribution $x_k^{(0)}$, learning rates $(\eta^{(t)}_k)$, and distance generating function $\psi_k$.}
\label{alg:smd_dynamics}
\end{algorithm}

Here, $\psi_k$ is a distance generating function defined and $C^1$ on $\simplexProd$, and $D_{\psi_k}$ is the Bregman divergence induced by $\psi_k$, defined as follows:
\begin{multline}
\label{eq:Bregman}
D_{\psi_k}(x_k, y_k) = \psi(x_k) - \psi(y_k) - \braket{\nabla \psi(y_k)}{x_k-y_k}
\end{multline}

\begin{assumption}
For all $k$, $\psi_k$ is strongly convex with respect to a reference norm $\| \cdot \|$. That is, there exists $\ell_{\psi_k} > 0$ such that for all $x_k, y_k \in\simplexProd$:
\[
D_{\psi_k}(x_k, y_k) \geq \frac{\ell_{\psi_k}}{2} \|x_k - y_k\|^2
\]
\end{assumption}
See Chapter 11 in~\cite{cesa2006prediction} for an account of the properties of Bregman divergences. We will further assume that the norm $\|\cdot\|$ decomposes into a sum of norms defined on each of the simplexes.
\begin{assumption} 
\label{ass:norm_sum}
The norm $\|\cdot\|$ on $\R^{\Pcal_1} \times \dots \times \R^{\Pcal_I}$ can be decomposed as follows:
\[
\|x_k\| = \sum_{i \in I} \|(x_k)_{\Pcal_i}\|
\]
\end{assumption}

Mirror descent is a general class of first-order optimization methods, used extensively both in convex optimization~\cite{nemirovski1983problem} and online learning~\cite{cesa2006prediction,bubeck2012regret}. In particular, projected gradient descent and entropic descent (a.k.a. the Hedge algorithm) are instances of the mirror descent method, for the appropriate choices of the distance generating function (see, for example,~\cite{Beck2003mirror}).

In our model, since each population is updating its distribution vector using mirror descent dynamics, we can write the joint update as follows
\begin{align}
(&x^{(t+1)}, \dots, x_K^{(t+1)}) \notag \\
&= \argmin_{x} \sum_{k} \braket{\ell^{(t)}}{x_k}_{\theta_k} + \sum_{k} \frac{1}{\eta_k^{(t)}} D_{\psi_k}(x_k, x^{(t)}_k) \notag \\
&= \argmin_{x} \braket{\nabla f(x^{(t)})}{x} + \sum_{k} \frac{1}{\eta_k^{(t)}} D_{\psi_k}(x_k, x^{(t)}_k) \notag \\
&= \argmin_{x} f(x^{(t)}) + \braket{\nabla f(x^{(t)})}{x - x^{(t)}} + D^{(t)}(x, x^{(t)})
\label{eq:joint_dynamics}
\end{align}
where the minimization is taken across $x$ in $(\simplexProd)^K$ and we used the expression of the gradient $\nabla f(x^{(t)})$, given in Proposition~\ref{prop:potential_gradient}, and defined:
\[
D^{(t)}(x, x^{(t)}) = \sum_{k} \frac{1}{\eta_k^{(t)}} D_{\psi_k}(x_k, x^{(t)}_k)
\]
The expression~\eqref{eq:joint_dynamics} can be interpreted as a local approximation of the potential function $f$: the first term $f(x^{(t)}) + \braket{\nabla f(x^{(t)})}{x - x^{(t)}}$ is simply the linear approximation of $f$ around $x^{(t)}$, and the second term $D^{(t)}(x, x^{(t)})$ is a strongly convex function which penalizes deviation from the previous iterate $x^{(t)}$. By this observation, one can think of the joint dynamics of all populations as implementing a stochastic mirror descent on the Rosenthal potential $f$.

\subsection{Suboptimality bounds on stochastic mirror descent}
We now review some guarantees of the stochastic mirror descent method. First, we need to make assumptions on the stochastic process $(\hat \ell^{(t)})$ and the distance generating functions $\psi_k$.

\begin{assumption}
\label{assumption:ell_hat}
Throughout the paper, we will assume that:
\begin{enumerate}
\item For all $t$, $\hat \ell^{(t)}$ is unbiased, that is, $\Exp{\hat \ell^{(t)} | \Fcal_{t-1}} = \ell^{(t)}$, where $(\Fcal_{t})$ is the natural filtration of the process $(\hat \ell^{(t)})$.
\item $\hat \ell^{(t)}$ is uniformly bounded in the squared dual norm, that is, there exists $L$ such that for all $t$:
\[
\Exp{\|\ell^{(t)}\|^2_*} \leq L
\]
where $\|\cdot\|_*$ is the dual norm defined as follows:
\[
\|\ell\|_* = \sup_{\|x\| \leq 1} \braket{x}{\ell}
\]
\item For all $k$, there exists $D_k$ such that $D_{\psi_k}$ is bounded on $\simplexProd$ by $D_k$.
\end{enumerate}
\end{assumption}

\begin{proposition}[Theorem~4 in~\cite{krichene2015SMD}]
\label{prop:SMD_subopt}
Suppose that each population $P_k$ follows a stochastic mirror descent dynamics as in Algorithm~\ref{alg:smd_dynamics}, and suppose that the learning rates are given by $\eta_k^{(t)} = c_k t^{-\alpha_k}$ with $c_k > 0$ and $\alpha_k \in (0, 1)$. Then for all $t \ge 1$, it holds that:
\begin{align*}
&\Exp{f(x^{(t)})} - f^\star \\
&\leq \parenth{1 + \sum_{\tau=1}^t \frac{1}{\tau} } 
\sum_{k = 1}^K \parenth{ \frac{1}{t^{1-\alpha_k}} \frac{D_k}{c_k} 
+ \frac{c_k L}{2\ell_{\psi_k}(1-\alpha_k)} \frac{1}{t^{\alpha_k}} }
\end{align*}
In particular, the system converges to the set of Nash equilibria in expectation, in the sense that $\Exp {f(x^{(t)})} \to f^\star$ at the rate $\Ocal(t^{-\bar \alpha} \log t )$ where $\bar \alpha = \min_k \min(\alpha_k, 1-\alpha_k)$.
\end{proposition}

\subsection{Sensitivity analysis of the stochastic mirror descent update}
In this Section, we study the sensitivity of the stochastic process $\hat \ell^{(t)}(x^{(t)})$ to changes in the private parameter $\theta$.

First, we consider how the flow allocations change due to a change in mass on some origin-destination pairs. In this case, we hold the observed loss vector $\hat \ell^{(t)}$ fixed and will invoke Corollary~\ref{cor:repeat} afterward.

\begin{proposition}
\label{prop:SMD_update_sensitivity}
Fix a loss vector $\hat \ell^{(t)}$ and consider the stochastic mirror descent update for population $P_k$
\[
x_k^{(t+1)}(\theta_k) = \argmin_{x_k} \braket{ \hat \ell^{(t)}}{x_k }_{\theta_k} + \frac{1}{\eta_k^{(t)}} D_{\psi_k}(x_k, x_k^{(t)})
\]
where the minimization is taken across $\simplexProd$. Here, $x^k$ is 
viewed as a function of the mass vector $\theta_k$. Then for all $\theta_k, \theta_k' \in \Rbb_+^I$:
\[
\|x_k^{(t+1)}(\theta_k) - x_k^{(t+1)}(\theta_k')\| \leq \frac{\eta_k^{(t)}\|\hat\ell^{(t)}\|_*}{\ell_{\psi_k}} \|\theta_k - \theta_k'\|_\infty
\]
\end{proposition}
\begin{proof}
The minimized function is differentiable on $\simplexProd$, and its gradient at $x_k$ is given by:
\[
((\theta_k)_i \hat \ell_{\Pcal_i}^{(t)})_{i \in I} + \frac{1}{\eta_k^{(t)}} \left[ \nabla \psi_k(x_k) - \nabla \psi_k(x_k^{(t)}) \right]
\]
To simplify the following expressions, we will use the following notation:
\begin{itemize}
\item $x_k^{(t+1)}(\theta_k)$ is denoted $x_k^{(t+1)}$, and $x^{(t+1)}_k(\theta_k')$ is denoted ${x'}^{(t+1)}$.
\item $g^{(t)}(\theta_k) = ((\theta_k)_i \hat \ell_{\Pcal_i}^{(t)})_{i \in I}$ 
\item $h_k^{(t)}(x_k) = \frac{1}{\eta_k^{(t)}} \left[ \nabla \psi_k(x_k) - \nabla \psi_k(x_k^{(t)}) \right]$
\end{itemize}
Then, by first-order optimality, we must have for all $x_k \in \simplexProd$:
\[
\braket{g^{(t)}(\theta_k) + h^{(t)}_k(x_k^{(t+1)})}{x_k - x_k^{(t+1)}} \geq 0
\]
In particular, for $x_k = {x'}^{(t+1)}$, we have:
\[
\braket{g^{(t)}(\theta_k) + h^{(t)}_k(x_k^{(t+1)})}{{x_k'}^{(t+1)} - x_k^{(t+1)}} \geq 0
\]
Permuting the roles of $\theta_k$ and $\theta_k'$, and summing the resulting inequalities, we have:
{\small\begin{multline}
\label{eq:sensitivity_proof_ineq}
\braket{g^{(t)}(\theta_k) - g^{(t)}(\theta_k')}{{x_k'}^{(t+1)} - x_k^{(t+1)}} \geq 
\\
\braket{h^{(t)}_k({x_k'}^{(t+1)}) - h^{(t)}_k(x_k^{(t+1)})}{{x_k'}^{(t+1)} - x_k^{(t+1)}}
\end{multline}}
Furthermore, we have by Cauchy-Schwartz:
\begin{align*}
&\braket{g^{(t)}(\theta_k) - g^{(t)}(\theta_k')}{{x_k'}^{(t+1)} - x_k^{(t+1)}} \\
&\leq \|g^{(t)}(\theta_k) - g^{(t)}(\theta_k')\|_* \|{x_k'}^{(t+1)} - x_k^{(t+1)}\|
\end{align*}
By strong convexity of $\psi_k$, we have:
{\small\begin{align*}
&\braket{h^{(t)}_k({x_k'}^{(t+1)}) - h^{(t)}_k(x_k^{(t+1)})}{{x_k'}^{(t+1)} - x_k^{(t+1)}} \\
&= \frac{1}{\eta_k^{(t)}}\braket{\nabla \psi_k({x_k'}^{(t+1)}) - \nabla \psi_k(x_k^{(t+1)})}{{x_k'}^{(t+1)} - x_k^{(t+1)}} \\
& \geq \frac{\ell_{\psi_k}}{\eta_k^{(t)}} \|{x_k'}^{(t+1)} - x_k^{(t+1)}\|^2
\end{align*}}
Combining these inequalities with~\eqref{eq:sensitivity_proof_ineq}, we have:
\begin{multline*}
\|g^{(t)}(\theta_k) - g^{(t)}(\theta_k')\|_* \|{x_k'}^{(t+1)} - x_k^{(t+1)}\| \\ \geq
\frac{ \ell_{\psi_k}}{\eta_k^{(t)}} \|{x_k'}^{(t+1)} - x_k^{(t+1)}\|^2
\end{multline*}
After simplification, this yields:
\[
\|{x_k'}^{(t+1)} - x_k^{(t+1)}\| \leq \frac{\eta_k^{(t)}}{\ell_{\psi_k}}\|g^{(t)}(\theta_k) - g^{(t)}(\theta_k')\|_*
\]
Finally, using the expression of $g^{(t)}(\theta_k) = ((\theta_k)_i \hat \ell_{\Pcal_i}^{(t)})_{i \in I}$, we have:
\begin{align*}
\|g^{(t)}(\theta_k) - g^{(t)}(\theta_k')\|_* 
&\leq \sum_{i \in I} \|\hat \ell^{(t)}_{\Pcal_i}\|_* |(\theta_k)_i - (\theta_k')_i| \\
&\leq \|\hat \ell^{(t)}\|_* \|\theta_k - \theta_k'\|_\infty
\end{align*}
which concludes the proof.
\end{proof}

We have bounded how much a change in the private parameter affects the distribution on paths. Now, we analyze how the flows are affected by changes in distribution.

We will use the notation $\phi(x ; \theta)$, which makes the dependence of edge flows on the parameter $\theta$ explicit. Also, $x^{(t+1)}(\theta)$ will be shorthand for $(x_1^{(t+1)}(\theta_1),\dots,x_K^{(t+1)}(\theta_K))$. Also, let $\|\cdot\|_a$ denote an arbitrary norm on the space of edge flows.

\begin{lemma}
\label{lem:flow_bnd}
For any $\adj(\theta,\theta')$, we have:
\begin{align*}
\| \phi(x^{(t+1)}(\theta); \theta) - &\phi(x^{(t+1)}(\theta'); \theta') \|_a \\
& \leq c A_x \left[A_\Delta +  A_\theta \frac{\eta_k^{(t)}\|\hat\ell^{(t)}\|_*}{\ell_{\psi_k}} \right]
\end{align*}
Here, $A_\theta$ is as given in Assumption~\ref{ass:bound_theta} and:
\[
A_x = \sup_{\|x_k\| \leq 1} \left\| \sum_{i = 1}^I M_i (x_{k})_{\Pcal_i} \right\|_a
\]
\[
A_\Delta = \sup_{x_k \in \simplexProd} \|x_k\|
\]
\end{lemma}

\begin{proof} Consider any $\adj(\theta,\theta')$. Note that $x_k^{(t+1)}(\theta) = x_k^{(t+1)}(\theta')$ for any $k \neq k^*$, since with the loss vector given, the update for population $k$ only depends on $\theta_k$.
\begin{align*}
\| \phi&(x^{(t+1)}(\theta); \theta) - \phi(x^{(t+1)}(\theta'); \theta') \|_a \\
& \leq \| \phi(x^{(t+1)}(\theta); \theta) - \phi(x^{(t+1)}(\theta); \theta') \|_a + \dots \\
& \qquad \| \phi(x^{(t+1)}(\theta); \theta') - \phi(x^{(t+1)}(\theta'); \theta') \|_a
\end{align*}

For the first term, since $\theta$ and $\theta'$ are adjacent:
{\small\begin{align*}
&\| \phi(x^{(t+1)}(\theta); \theta) - \phi(x^{(t+1)}(\theta); \theta') \|_a \\
& = \left\| \sum_{i = 1}^I (\theta_{k^*})_i M_i (x_{k^*}^{(t+1)}(\theta))_{\Pcal_i} - 
\sum_{i = 1}^I (\theta_{k^*}')_i M_i (x_{k^*}^{(t+1)}(\theta))_{\Pcal_i} \right\|_a \\
& \leq \left\| \sum_{i = 1}^I |(\theta_{k^*} - \theta_{k^*}')_i| M_i (x_{k^*}^{(t+1)}(\theta))_{\Pcal_i} \right\|_a \\
& \leq c\left\| \sum_{i = 1}^I  M_i (x_{k^*}^{(t+1)}(\theta))_{\Pcal_i} \right\|_a \leq c A_x A_\Delta
\end{align*}}%
For the second term, we invoke Proposition~\ref{prop:SMD_update_sensitivity}:
{\small\begin{align*}
&\| \phi(x^{(t+1)}(\theta); \theta') - \phi(x^{(t+1)}(\theta'); \theta') \|_a \\
& = \left\| \sum_{i = 1}^I (\theta_{k^*}')_i M_i (x_{k^*}^{(t+1)}(\theta))_{\Pcal_i} - \sum_{i = 1}^I (\theta_{k^*}')_i M_i (x_{k^*}^{(t+1)}(\theta'))_{\Pcal_i} \right\|_a \\
& = \left\| \sum_{i = 1}^I (\theta_{k^*}')_i M_i \left[ (x_{k^*}^{(t+1)}(\theta))_{\Pcal_i}- (x_{k^*}^{(t+1)}(\theta'))_{\Pcal_i} \right] \right\|_a \\
& \leq A_\theta \left\| \sum_{i = 1}^I M_i \left[ (x_{k^*}^{(t+1)}(\theta))_{\Pcal_i}- (x_{k^*}^{(t+1)}(\theta'))_{\Pcal_i} \right] \right\|_a \\
& \leq A_\theta A_x \left\|x_{k^*}^{(t+1)}(\theta)-x_{k^*}^{(t+1)}(\theta')\right\|
 \leq c A_\theta A_x \frac{\eta_{k^*}^{(t)}\|\hat\ell^{(t)}\|_*}{\ell_{\psi_{k^*}}}
\end{align*}}%
As desired.
\end{proof}

We have bounded the effect of a change in the private parameter on the flows. 
Thus, we can state the sensitivity of the loss vector at time $t+1$ due to a small differential in the private parameter $\theta$, when the observed loss vector at time $t$ is held fixed.

\begin{theorem}
\label{theorem:sens_loss}
	\emph{Sensitivity of the loss function: } 
	For any $\adj(\theta,\theta')$:
	\begin{align*}
	\| \ell(x^{(t+1)}(\theta); \theta, x^{(t)}, \hat \ell^{(t)}) - \ell(x^{(t+1)}(\theta'); \theta', x^{(t)}, \hat \ell^{(t)}) \| \\
	\leq c A_\ell A_x \left[A_\Delta +  A_\theta\frac{\max_{k \in [K]}(\eta_{k}^{(t)}) \|\hat\ell^{(t)}\|_*}{\min_{k \in [K]} (\ell_{\psi_{k}})}\right]
	\end{align*}
	Here, $A_x, A_\Delta,$ and $A_\theta$ are as defined in Assumption~\ref{ass:bound_theta} and Lemma~\ref{lem:flow_bnd}, and $A_\ell$ denotes the Lipschitz constant of the function $\ell : \phi \mapsto \ell(\phi)$ with respect to the norm $\|\cdot\|_a$ on the domain and $\|\cdot\|$ on the codomain.
\end{theorem}

Note that the sensitivity of $\ell^{(t+1)}$ depends on $t$ through the learning rate $\eta_k^{(t)}$.


\section{Differential privacy of the routing game}
\label{sec:diff_priv_route}


In this Section, we use results from the previous sections to give privacy guarantees on the routing game when the loss vectors are observed with Gaussian noise.

Also, recall that the Gaussian mechanism preserves $(\eps,\delta)$ differential privacy, and the privacy values depend on the variance of the mechanism and the sensitivity of the function. 
At each iteration $t$, we suppose that the populations observe $\hat \ell^{(t)} = \ell(x^{(t)}) + Z_t$ where $(Z_t)_p \simiid \gauss(\sigma^2)$.

We offer a couple of different interpretations of this mechanism. The first is that the data collector adds Gaussian noise before releasing this data to the populations. For example, the Department of Transportation might choose to add noise before transmitting the measurements from inductive-loop detectors in the road for privacy purposes. The second interpretation is that each driver experiences a perturbed version of the nominal loss when driving along the road, and when a population aggregates these perturbations, they obey a central limit theorem and look roughly normal in distribution.

First, we observe that for each path $p$, since the loss function $\ell_p$ is continuous on the compact set $\parenth{\simplexProd}^K$, it is bounded. Therefore, there exists $M > 0$ such that for all $x \in \parenth{\simplexProd}^K$, $\|\ell(x)\|_\infty \leq M$. 

\begin{theorem}
\label{theorem:eps_delta_bnd}
After $T$ iterations, the mapping $\theta \mapsto (\hat \ell^{(1)}, \dots, \hat \ell^{(T)})$ is $(\eps,\delta)$ differentially private, where:
\[
\eps = \sum_{t = 1}^T \eps_t
\qquad
\delta = \sum_{t = 1}^T \exp\left[ \sum_{t' = t+1}^T \eps_{t'} \right]\delta_t + \delta'
\]
Here, $a$ is any positive constant and $\delta', \eps_t, \delta_t$ are any constants that satisfy the following constraints:
\[
1 - \delta' = (1 - 2\exp(-a^2/2\sigma^2))^{T\sum_{i = 1}^I |\Pcal_i|}
\]
\begin{align*}
\eps_t & > \frac{ c A_\ell A_x (2\ln(1.25/\delta_t))^{1/2} }{\sigma^2} \times \\
& \left[A_\Delta +  A_\theta \frac{\max_{k \in [K]}(\eta_{k}^{(t)}) ({\sum_{i = 1}^I |\Pcal_i|})^{1/2}(M+a)}{\min_{k \in [K]} (\ell_{\psi_{k}})} \right]
\end{align*}
$A_x, A_\Delta,A_\theta,$ and $A_\ell$ are as defined in Assumption~\ref{ass:bound_theta}, Lemma~\ref{lem:flow_bnd}, and Theorem~\ref{theorem:sens_loss}.
\end{theorem}

\begin{proof}
We can invoke the Chernoff bound and see that $P(|(Z_t)_p| > a) \leq 2\exp(-a^2/2\sigma^2)$. It follows that the event $E = \{\|Z_t\|_\infty \leq a \text{ for all } t\}$ holds with at least probability $(1 - 2\exp(-a^2/2\sigma^2))^{T\sum_{i = 1}^I |\Pcal_i|}$. On $E$, we have that $\|\hat \ell^{(t)}\|_2 \leq ({\sum_{i = 1}^I |\Pcal_i|})^{1/2} \|\hat \ell^{(t)}\|_\infty \leq ({\sum_{i = 1}^I |\Pcal_i|})^{1/2}(M+a)$ a.s.

Invoking Theorem~\ref{theorem:sens_loss}, we can see that, on $E$:
\begin{align*}
\Delta_2& \ell^{(t+1)} \leq \\
&
c A_\ell A_x \left[A_\Delta +  A_\theta \frac{\max_{k \in [K]}(\eta_{k}^{(t)}) ({\sum_{i = 1}^I |\Pcal_i|})^{1/2}(M+a)}{\min_{k \in [K]} (\ell_{\psi_{k}})} \right]
\end{align*}

Thus, invoking Proposition~\ref{prop:gauss_mech}, Corollary~\ref{cor:repeat}, and Lemma~\ref{lemma:tail_event} yields our desired result.
\end{proof}

Note that $a$ can be chosen to be any positive constant, and, in effect, provides a trade-off between the $\eps$ and the $\delta$ parameters.

\section{Numerical Example}
\label{sec:numerical}


Consider the routing game played on the network in Figure~\ref{fig:network}, with the following populations:
\begin{enumerate}
\item Population $P_1$ has mass vector $\theta_1 = (1, 0)$, and follows stochastic mirror descent dynamics with learning rates $\Ocal(t^{-.5})$.
\item Population $P_2$ has mass vector $\theta_2 = (.2, 1.2)$, and follows stochastic mirror descent dynamics with learning rates $\Ocal(t^{-.2})$.
\end{enumerate}

\def\scale{.47}
\begin{figure}[h]
\centering
\includegraphics[width=.3\textwidth]{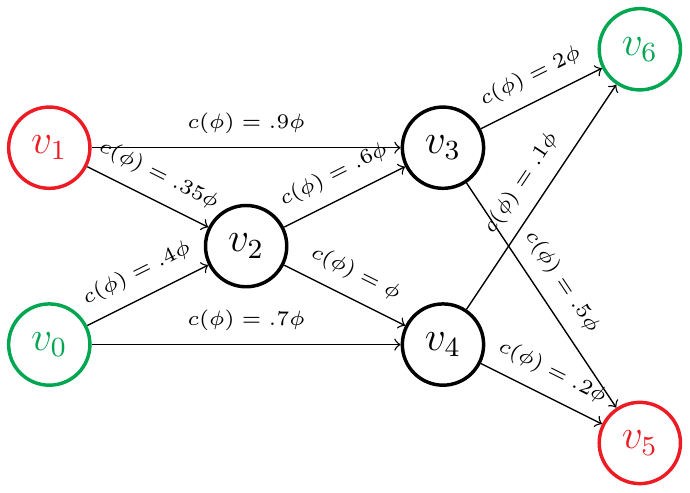}
\caption{Example network with two origin-destination pairs: $(v_0, v_6)$ and $(v_1, v_5)$.}
\label{fig:network}
\end{figure}


The losses are taken to be linear. The resulting path loss functions are bounded by $M = 2$. We simulate the game for $T = 200$ iterations, with Gaussian noise with standard deviation $\sigma \in \{.01, .1, .4\}$. 


\begin{figure}[h!]
\centering
\includegraphics[page=1,width=.4\textwidth]{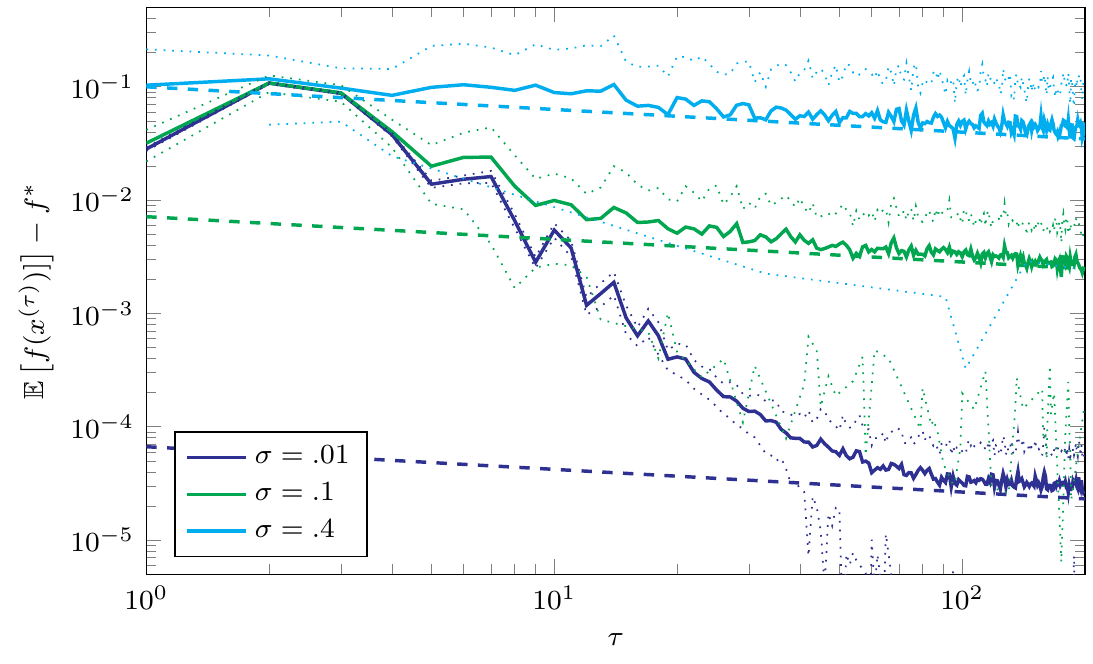}
\caption{Potential function values $f(x^{(\tau)})$ as a function of the iteration $\tau$, for different values of $\sigma$. The solid and dotted lines show, respectively, the average and the standard deviation over 150 runs of the simulation. The dashed lines show the $\tilde \Ocal(t^{-.2})$ asymptotic rate predicted by Proposition~\ref{prop:SMD_subopt}.}
\label{fig:potentials}
\end{figure}

\def\scale{.36}
\begin{figure}[h!]
\centering
\includegraphics[page=2,width=\scale\textwidth]{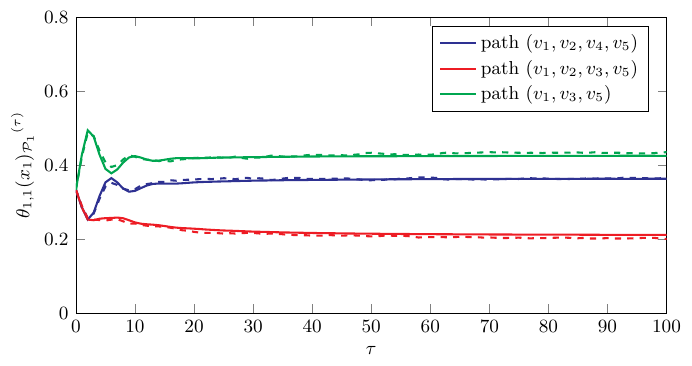}
\includegraphics[page=3,width=\scale\textwidth]{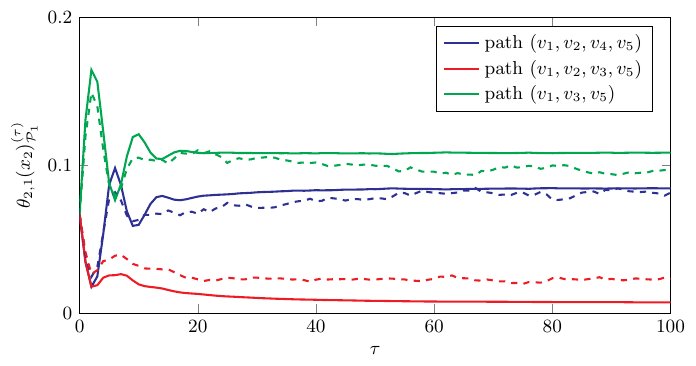}
\includegraphics[page=4,width=\scale\textwidth]{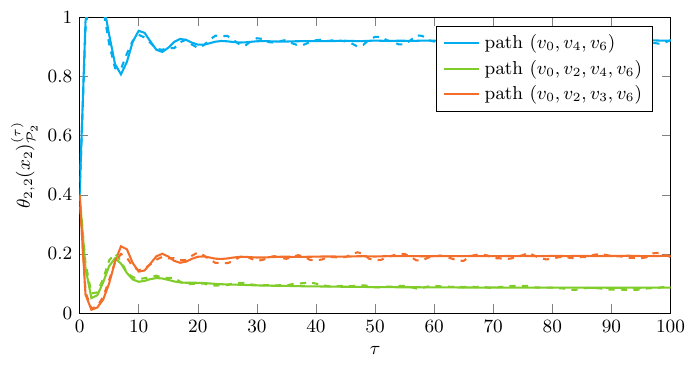}
\caption{Path flows for each population, averaged over 150 runs, for $\sigma = .01$ (solid lines) and $\sigma = .4$ (dashed lines)}
\label{fig:flows}
\end{figure}

Figure~\ref{fig:potentials} shows the values of the potential function for the different values of $\sigma$. The asymptotic rate is consistent with $\tilde \Ocal(t^{-\min(\alpha_1, \alpha_2)}) = \tilde \Ocal(t^{-.2})$ rate predicted by Proposition~\ref{prop:SMD_subopt}. The variance of the noise $\sigma^2$ significantly affects the value of the expected potential. 
The effect of $\sigma$ can also be observed in Figure~\ref{fig:flows}, which shows the path flows for both populations, for $\sigma \in \{.01, .4\}$. Besides the effect of the noise level, we also observe that because the learning rates of population $P_2$ have a slower decay rate, its updates are more aggressive, which is reflected in the trajectories of its path flows.

Additionally, we consider the differential privacy of these observable traffic flows. Applying Theorem~\ref{theorem:eps_delta_bnd}, we plot the differential privacy values as a function of the number of iterations in Figure~\ref{fig:eps_delta}. Generally, we are able to mask a small amount of population flow, but should $c$ grow too large, the bounds quickly become trivial, i.e. $\delta = 1$. Furthermore, this value at which we can no longer meaningfully guarantee privacy can be thought of as the rate at which populations must shift origin-destination pairs to retain some level of privacy guarantee.
\begin{figure}[h!]
\centering
\includegraphics[page=7,width=.5\columnwidth]{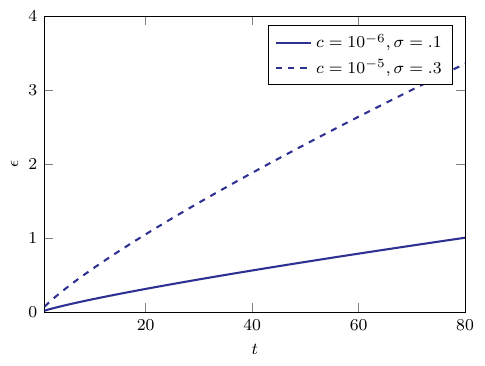}%
\includegraphics[page=8,width=.5\columnwidth]{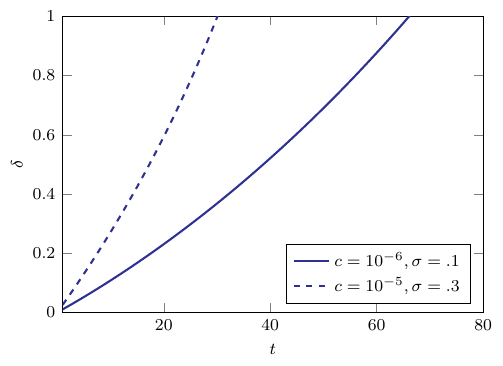}
\caption{A plot of the values of $\epsilon, \delta$ for which differential privacy holds, as a function of $t$, the number of iterations. Here, $(c, \sigma)$ are taken to be $(10^{-6}, .1)$ then $(10^{-5}, .3)$, and $a$ is taken to be $2$. For larger values of $c$, the privacy guarantees are only meaningful for shorter periods of time.}
\label{fig:eps_delta}
\end{figure}

\vspace{-.1in}
\section{Conclusion}
\label{sec:conclusion}


In this paper, we considered the privacy of the origins and destinations of drivers when the nominal traffic losses are observable with Gaussian noise. Considering a general online learning model based on stochastic mirror descent, and noting that the routing game is a potential game, we can think of the dynamics of drivers as optimizing the Rosenthal potential.

We analyzed the sensitivity of each update step as a function of the masses for each origin-destination pair, which allowed us to bound the influence of this private information on the observable traffic losses. Additionally, we provided bounds on the convergence rates for different levels of noise, which provides insight into the relationship between how long it takes traffic flows to settle at equilibrium and how much is revealed by these observable traffic costs.


\bibliographystyle{IEEEtran}
\bibliography{root}

\end{document}